\documentclass[12pt,A4paper,peerreview,draftcls]{IEEEtran}
\usepackage{hyperref}
\usepackage{graphicx}
\usepackage{ifthen, calc, amssymb, amsmath, amsfonts, subfigure}

\def\calA{{\cal A}}
\def\calB{{\cal B}}
\def\calC{{\cal C}}

\def\calP{{\cal P}}

\def\calR{{\cal R}}
\def\calS{{\cal S}}

\def\calU{{\cal U}}

\def\calX{{\cal X}}

\def\tR{{\tilde R}}
\def\tU{{\tilde U}}
\def\tX{{\tilde X}}

\newcommand{\given}{\mid}

\newtheorem{theorem}{Theorem}

\newtheorem{corollary}[theorem]{Corollary}
\newtheorem{proposition}[theorem]{Proposition}
\newtheorem{definition}{Definition}
\newtheorem{example}{Example}

\begin{document}

\title{Error Free Perfect Secrecy Systems}
\author{
\thanks{The material in this paper was presented in part
at the IEEE International Symposium on Information Theory,
St. Petersburg, July, 2011.  The work of S.-W. Ho was supported by the Australian
    Research Council under an Australian Postdoctoral Fellowship as
    part of Discovery Project DP1094571. The work of T. H. Chan and
    A. Grant was also supported in part by ARC Discovery Project
    DP1094571. }
\authorblockN{Siu-Wai Ho, Terence H. Chan, Alex Grant and Chinthani Uduwerelle}\\
\authorblockA{Institute for Telecommunications Research\\
University of South Australia
}}

\maketitle
\begin{abstract}
  Shannon's fundamental bound for perfect secrecy says that the
  entropy of the secret message cannot be larger than the entropy of
  the secret key initially shared by the sender and the legitimate
  receiver. Massey gave an information theoretic proof of this result,
  however this proof does not require independence of the key and
  ciphertext. By further assuming independence, we obtain a tighter
  lower bound, namely that the key entropy is not less than the
  logarithm of the message sample size in any cipher achieving perfect
  secrecy, even if the source distribution is fixed. The same bound
  also applies to the entropy of the ciphertext. The bounds still hold
  if the secret message has been compressed before encryption.

  This paper also illustrates that the lower bound only gives the
  minimum size of the pre-shared secret key. When a cipher system is
  used multiple times, this is no longer a reasonable measure for the
  portion of key consumed in each round. Instead, this paper proposes
  and justifies a new measure for key consumption rate. The existence
  of a fundamental tradeoff between the expected key consumption and
  the number of channel uses for conveying a ciphertext is
  shown. Optimal and nearly optimal secure codes are designed.
\end{abstract}


\newpage
\section{Introduction}
\label{se:intro}
Cipher systems with \emph{perfect secrecy} were studied by Shannon in
his seminal paper~\cite{shannon1949communication} (see
also~\cite{massey1992}). With reference to Figure~\ref{fig:encdec}, a
cipher system is defined by three components: a source message $U$, a
ciphertext $X$ and a key $R$. The key is secret common randomness
shared by the sender and the legitimate receiver. The sender encrypts
the message $U$, together with the key $R$, into the ciphertext
$X$. This ciphertext will be transmitted to the legitimate receiver
via a public channel. A cipher system is \emph{perfectly secure}, or
equivalently, satisfies a perfect secrecy constraint if the message
$U$ and the ciphertext $X$ are statistically independent, $I(U; X) =
0$. In this case, an adversary who eavesdrops on the public channel
and learns $X$ (but does not have $R$) will not be able to infer any
information about the message $U$. On the other hand, the legitimate
receiver decrypts the message $U$ from the received ciphertext $X$
together with the secret key $R$. A cipher system is \emph{error-free}
(i.e., the probability of decoding error is zero) if $H(U \given X R)
= 0$.
\begin{figure}[htbp]
  \begin{center}
    \includegraphics[scale=0.8]{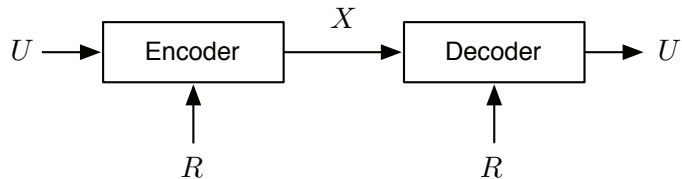} \caption{\label{fig:encdec}
      A cipher system.}
  \end{center}
\end{figure}

By considering a \emph{deterministic cipher}, where $X$ is a
deterministic function of $R$ and $U$, Shannon showed that the number
of messages is equal to the number of possible ciphertexts, and that
the number of different keys is not less than the number of
messages~\cite[p. 681]{shannon1949communication},
\begin{equation*}
|\calX| = |\calU| \leq |\calR|,
\end{equation*}
where $\calX$, $\calU$ and $\calR$ are the respective supports of $X$,
$U$ and $R$.  In order to design a perfectly secure cipher system
protecting a source with \emph{unknown source distribution} $P_U$,
Shannon argued that
\begin{equation}
  H(R) \geq \log |\calU| \geq H(U) .  \label{eq:shannonB0}
\end{equation}
He also made an important observation~\cite[p. 682]{shannon1949communication} that
\begin{quote}
  ``\emph{the amount of of uncertainty we can introduce into the
    solution cannot be greater than the key uncertainty}"
\end{quote}
In other words,
\begin{equation}
  H(R) \geq H(U).   \label{eq:shannonB}
\end{equation}
Massey~\cite{massey1992} called~\eqref{eq:shannonB} Shannon's
fundamental bound for perfect secrecy, and gave an information
theoretic proof for this result.  It is important to note that
Massey's proof~\cite{massey1992} does not require $U$ and $R$ to be
statistically independent.

Now, suppose $U$ and $R$ are indeed independent (which is common in
practice).  Our first main result, Theorem~\ref{th:BoundonXR}
improves~\eqref{eq:shannonB}, showing that for any source distribution
$P_U$,
\begin{equation}
  P_R(r) \leq |\calU|^{-1}, \quad \forall r. \label{eq:3}
\end{equation}
As a consequence, we prove  that for any
cipher achieving perfect secrecy, the logarithm of the message sample
size cannot be larger than the entropy of the secret key,
\begin{equation}
  H(R) \geq \log |\calU|.  \label{eq:firstresult}
\end{equation}
Comparing with the first inequality in~\eqref{eq:shannonB0}, we see
that~\eqref{eq:firstresult} is valid even if the source distribution
$P_U$ is fixed and known.

This paper is based on the model in Fig.~\ref{fig:encdec}.  Despite
its apparent simplicity, this is the most general encoder possible,
and covers many interesting special cases. For example, suppose the
distribution of $U$ is non-uniform.  One may expect that the optimal
encoder will operate according to Fig.~\ref{fig:comencry}, by first
compressing $U$ and then encrypting the compressed output.
\begin{figure}[htbp]
  \begin{center}
    \includegraphics[scale=0.8]{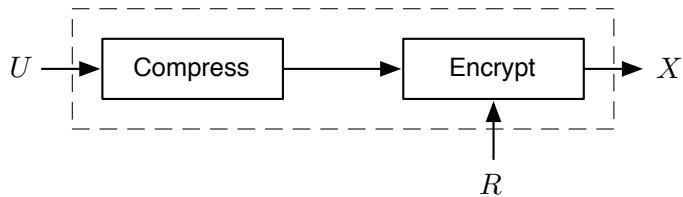} \caption{\label{fig:comencry}
      Compression before encryption.}
  \end{center}
\end{figure}

Roughly speaking, compression converts the source into a sequence of
independent and identically distributed (i.i.d.) symbols.
Theoretically, this can maximize the adversary's decoding error
probability in some systems~\cite[Theorem~3]{ReliabilityISIT}.
Practically, the compressed output has a smaller file size and hence
seem to require less key for encryption. This approach of compression
before encryption was also proposed by Shannon
\cite[p. 682]{shannon1949communication}.  In fact, Shannon believed
that, after removing redundancy in the source,
\begin{quote}
  ``\emph{a bit of key completely conceals a bit of message
    information}".
\end{quote}
However a separated compression before encryption model is a special
case of our more general model in Fig.~\ref{fig:encdec}. To certain
extent, our model can be viewed as joint compression-encryption
coding. Naturally, our results also apply to models such as
Fig.~\ref{fig:comencry}, for which we will later prove
\begin{equation*}
  H(X) \geq \log |\calU|.
\end{equation*}
This result, together with~\eqref{eq:firstresult}, in fact suggest
that compression before encryption may not be useful if both perfect
secrecy and error-free decoding are required.

Another major contribution of this paper is the introduction of a new
concept of \emph{expected key consumption} $I(R; U X)$.  Previously in
the literature, the amount of key required in a cipher system has been
measured by the entropy of the common secret key.  We will argue
in this paper that $H(R)$ is only valid for measuring the
\emph{initial key requirement}, by which we mean the amount of secret
randomness that must be shared between the sender and the legitimate
receiver, prior to transmission of the ciphertext. Instead, key
consumption should be measured by $I(R; U X)$. This new measure offers
more insights, and in the second part of this paper, we will design
efficient cipher system that can be used multiple times, where $I(R; U
X)$ is one of the system parameters to be optimised.

Besides expected key consumption, we also want to minimize the
\emph{number of channel uses} required to transmit the ciphertext $X$
from the source to the legitimate receiver.  Naturally, we can encode
the ciphertext $X$ using a Huffman code~\cite{bk:Cover}.  Let $\lambda(X)$ be the codeword length.  In this
case, the expected codeword length $\mathbf{E}[\lambda(X)]$ satisfies $H(X) \leq
\mathbf{E}[\lambda(X)] \leq H(X) + 1$.  Note that for two random variables $X$ and
$X'$, it is possible that $H(X) < H(X')$, but $\mathbf{E}[\lambda(X)] > \mathbf{E}[\lambda(X)]$.
One example is when $P_{X} = (0.3, 0.23, 0.2, 0.17, 0.1)$ and $P_{X'}
= (0.25,0.25,0.25,0.15, 0.1)$.  However, we still use $H(X)$ instead
of $\mathbf{E}[\lambda(X)]$ as a measure for the number of channel uses required in a
cipher system for two reasons: first, $H(X)$ is a lower bound for
$\mathbf{E}[\lambda(X)]$ and in fact a very good estimate for $\mathbf{E}[\lambda(X)]$; second, the
problem itself is more tractable when using $H(X)$, instead of
$\mathbf{E}[\lambda(X)]$.

We will show that there exists a fundamental tradeoff between the
expected key consumption and the number of channel uses. In fact, if
the source distribution is not uniform, then the minimum expected key
consumption and the minimum number of channel uses cannot be
simultaneously achieved.  We will also show that code design achieving
minimum expected key consumption depends on whether the source
distribution $P_U$ has irrational probability masses or not. Optimal
code will be proposed for $P_U$ which has only rational probability
masses.

\emph{Organization:} In Section~\ref{se:SingleUse}, we consider
one-shot systems, where there is a single message to be securely
transmitted. We formalize the system model, and new bounds on $H(R)$
and $H(X)$ will be derived.  In Section~\ref{se:MultUse}, we will
consider the case where cipher system is used multiple times.  New
system parameters including $I(R; UX)$ will be defined and justified.
Section~\ref{se:Tradeoff} will focus on two regimes corresponding to
minimal expected key consumption and minimal number of channel uses.
The existence of a fundamental non-trivial tradeoff will be
illustrated.  In Section~\ref{se:ComEnc}, the performance of
compression-before-encryption will be evaluated.

\emph{Notation.} Random variables are denoted by capital letters, e.g.
$X$, and their particular realizations are denoted by small letters,
$x$.  Supports of random variables are denoted by calligraphic
letters, $\calX$.

\section{Key Requirements for One-Shot Ciphers}\label{se:SingleUse}

\begin{definition}[Error free perfect secrecy system]
\label{de:eps}
A cipher system $(R, U, X)$ is called an \emph{Error-free
  Perfect-Secrecy} (EPS) system if
\begin{align}
  I(U; X) &= 0, \label{eq:IUX0} \\
  H(U\given R X) &= 0, \label{eq:HURX0}\\
  I(U; R) &= 0. \label{eq:IUR0}
\end{align}
\end{definition}
Here,~\eqref{eq:IUX0} ensures perfect secrecy, via independence of the
ciphertext $X$ and source message $U$. An eavesdropper learning $X$
can infer no information about the message $U$.  The
constraint~\eqref{eq:HURX0} ensures that
the receiver can reconstruct $U$ from $R$ and $X$ without error. Finally~\eqref{eq:IUR0} requires that the shared secret key $R$ is
independent of the message $U$.

The constraints~\eqref{eq:IUX0} and~\eqref{eq:HURX0} were originally
used in~\cite{massey1992} to prove Shannon's fundamental
bound~\eqref{eq:shannonB} for perfect secrecy.  The only additional
constraint in Definition~\ref{de:eps} is~\eqref{eq:IUR0}.  In
practice, $R$ is usually shared prior to the independent generation of
the message $U$.  This is a strong practical motivation
for~\eqref{eq:IUR0}.  Furthermore, Definition~\ref{de:eps} admits the
general case of probabilistic encoding.  For the receiver, it is
however sufficient to consider deterministic decoding since
by~\eqref{eq:HURX0}, $U$ is a function of $R$ and $X$.  In other
words, there exists a decoding function $g$ such that
\begin{equation}
  P_{URX}(u, r, x) = P_{RX}(r, x) \mathbf{1} \{u = g(r, x)\}. \label{eq:cov1}
\end{equation}

\begin{theorem}[Lower bounds on $H(X)$ and $H(R)$]
  \label{th:BoundonXR}
  Let $(R, U, X)$ be an error free prefect secrecy system, satisfying
  \eqref{eq:IUX0} -- \eqref{eq:IUR0} according to
  Definition~\ref{de:eps}, and suppose $P_U$ is known. Then
  \begin{align}
    \max_{x\in\calX} P_{X}(x) & \le   |\calU|^{-1}, \label{eq:thm1px} \\
    \intertext{and}
    \max_{r\in\calR} P_{R}(r) & \le   |\calU|^{-1}, \label{eq:thm1pr}
  \end{align}
  where $\calU$ is the support of the message $U$.  Consequently,
  \begin{equation}
    \log |\calU| \le H(X) , \label{eq:thm1hx}
  \end{equation}
  with equality if and only if $P_X(x) = |\calU|^{-1}$ for all $x \in
  \calX$.  Also,
  \begin{equation}
    \log |\calU| \le H(R), \label{eq:thm1hr}
  \end{equation}
  with equality if and only if $P_R(r) = |\calU|^{-1}$ for all $r \in
  \calR$.  If the source distribution is not uniform, $H(X)$ and
  $H(R)$ are strictly greater than $H(U)$.
\end{theorem}

\begin{proof}
  For any $x \in \calX$,
  \begin{align}
    |\calU|P_{X}\left(x\right)
    &= \sum_u P_{X}\left(x\right) \label{eq:cov-1} \\
    &= \sum_u P_{X\given U}\left(x \given u\right)  \label{eq:cov0}\\
    &= \sum_u \sum_{r: P_{URX}\left(u, r, x\right) > 0}
       \frac{P_{URX}\left(u, r, x\right)}{ P_{U}\left(u\right)} \\
    &= \sum_u \sum_{r: P_{URX}\left(u, r, x\right) > 0}
       \frac{P_{RX}\left(r, x\right) \mathbf{1}\, \{u = g\left(r, x\right)\}}
       {P_{U}\left(u\right)}  \label{eq:cov01}\\
    &= \sum_{r: P_{RX}\left(r, x\right) > 0}
       \frac{P_{RX}\left(r, x\right) } {P_{U}\left(g\left(r,x\right)\right)} \\
    &= \sum_{r: P_{RX}\left(r, x\right) > 0} P_{RX}\left(r, x\right)
       \frac{P_{X \given UR}\left(x\given g\left(r,x\right), r\right)\,
        P_{R}\left(r\right)}{P_{URX}\left(g\left(r,x\right), r, x\right)}
     \label{eq:cov2}\\
    &= \sum_{r: P_{RX}\left(r, x\right) > 0}
       P_{X\given UR}\left(x\given g\left(r,x\right), r\right)
       P_{R}\left(r\right) \label{eq:cov3}\\
    &\leq \sum_{r: P_{RX}\left(r, x\right) > 0}  P_{R}\left(r\right)\\
    &\leq 1, \label{eq:cov4}
  \end{align}
  where \eqref{eq:cov0}, \eqref{eq:cov01}, \eqref{eq:cov2} and
  \eqref{eq:cov4} follow from \eqref{eq:IUX0}, \eqref{eq:cov1},
  \eqref{eq:IUR0} and \eqref{eq:cov1}, respectively.  This
  establishes~\eqref{eq:thm1px}.

  Let $P_B$ be a uniform distribution with support $\calU$. Since
  $P_{X}$ is always majorized\footnote{A good introduction to
    majorization theory can be found in \cite{bk:majorization}. In
    this proof, we just need the definition of ``majorized by" which
    can also be found in \cite[Definition~1]{FanoJ}} by $P_B$ from
  \eqref{eq:thm1px}, \cite[Theorem~10]{FanoJ} shows that
  \begin{align}
    H\left(P_X\right) &\geq H\left(P_B\right) + D\left(P_B \Vert P_X\right) \\
    &\geq H\left(P_B\right)  \label{eq:cov5}\\
    &= \log |\calU|,
  \end{align}
  and hence \eqref{eq:thm1hx} is verified.  Note that
  \cite[Theorem~10]{FanoJ} can still be applied even if $X$ may be
  defined on a countably infinite alphabet.  If $H(P_X) = \log
  |\calU|$, equality in~\eqref{eq:cov5} holds so that $P_X \equiv
  P_B$.  Finally,~\eqref{eq:thm1pr} and~\eqref{eq:thm1hr} follow from
  the symmetric roles of $X$ and $R$ in
  \eqref{eq:IUX0} -- \eqref{eq:IUR0}.
\end{proof}

\begin{corollary}
  \label{th:noEPS}
  No error free perfect secrecy system can be constructed if the
  source message $U$ has a countably infinite support or a support
  with unbounded size.
\end{corollary}
\begin{proof}
  Assume in contradiction that an EPS system exists for a source
  message $U\sim P_U$ with countably infinite support, $|\calU|
  = \infty$.  Note that \eqref{eq:cov-1} -- \eqref{eq:cov4} are still
  valid in this case. However, the conclusion that $|\calU|P_{X}(x)
  \leq 1$ for any $x \in \calX$ contradicts $|\calU| = \infty$.
\end{proof}

The following three remarks emphasize some of the (perhaps unexpected)
consequences of Theorem~\ref{th:BoundonXR}.
\begin{enumerate}
\item One could naturally expect that $H(U)$ is the critical quantity
  setting a lower bound on $H(R)$ and $H(X)$.  However,
  Theorem~\ref{th:BoundonXR} shows that $H(R)$ and $H(X)$ can be
  arbitrarily large, as long as the size of the support of $U$ is also
  arbitrarily large, \emph{even when $H(U)$ is small}.

\item One may further expect that $\log |\calU| \leq H(R)$ is tight
  only if the source distribution $P_U$ is unknown.
  However,~\eqref{eq:thm1hx} and~\eqref{eq:thm1hr} show that fixing
  $P_U$ does not reduce the lower bounds on either the initial key
  requirement, or the number of channel uses required to convey the
  ciphertext.

\item If the source message $U$ is defined on a countably infinite
  alphabet, it is not possible to design an error free perfect secrecy
  system (Corollary~\ref{th:noEPS}).  Therefore, if a cipher system is
  required for such a source, at least one of the constraints
  \eqref{eq:IUX0} -- \eqref{eq:IUR0} must be relaxed.
\end{enumerate}

The following example compares Shannon's fundamental
bound~\eqref{eq:shannonB} with Theorem~\ref{th:BoundonXR}.  It also
illustrates that the quantity $H(R)$ is insufficient for determination
of the requirements on the secret key $R$.
\begin{example}
Suppose $P_U = (0.3, 0.3, 0.3, 0.1)$ so that $H(U) = 1.895$ bits and
$\log|\calU| = 2$ bits.
\begin{enumerate}
\item Consider $R$ chosen independently of $U$ according to $P_R =
  (0.4, 0.2, 0.2, 0.2)$ so that $H(R) = 1.922$ bits and $H(U) < H(R) <
  \log|\calU|$.  Although $P_R$ satisfies Shannon's fundamental
  bound~\eqref{eq:shannonB}, Theorem~\ref{th:BoundonXR}, in
  particular~\eqref{eq:thm1hr}, shows this choice of key $R$ is
  insufficient to achieve error free perfect secrecy.

\item Consider $P_R = (0.4, 0.15, 0.15, 0.15, 0.15)$ so that $H(R) =
  2.171$ bits and $H(U) < \log|\calU| < H(R)$.  However, this choice
  of key $R$ is insufficient for error free perfect secrecy, since
  from~\eqref{eq:thm1pr}, $\max_r P_R(r) = 0.4 > 0.25 = |\calU|^{-1}$.
\end{enumerate}
\end{example}

Theorem~\ref{th:BoundonXR} not only applies to systems of the form
shown in Fig.~\ref{fig:encdec} (which includes Fig.~\ref{fig:comencry}
as a special case), but also to multi-letter variations.  For example,
we can accumulate $n$ symbols from the source $(U_1, U_2, \ldots,
U_n)$ and treat these $n$ symbols together as one super-symbol $U$. It
is reasonable to consider finite $n$ because practical systems have
only finite resources to store the super-symbol.  Unless the source
has some special structure, the distribution of $U$ cannot be uniform
for any $n$ if the $U_i$ are not uniform.  For example, if the source
is stationary and memoryless, accumulating symbols will only make
$H(X)$ and $H(R)$ grow with $n \log |\calU|$.

One may argue that the coding rate of $H(X)$ could be reduced because
the sender and receiver share the same side information $R$ and $I(X;
R) > 0$ is possible.  In other words, a compressor may be appended to
the encoder in Fig.~\ref{fig:encdec} in order to reduce the size of
the ciphertext. This configuration is shown in
Fig.~\ref{fig:HXgivenR}.  However, we cannot simply apply the results
from source coding with side information here, because the ciphertext
still needs to satisfy the security constraint.  If the new output $Y$
satisfies the perfect secrecy and zero-error constraints, $I(U; Y) =
H(U|RY) = 0$, then $(R, U, Y)$ in Fig.~\ref{fig:HXgivenR} is simply
another EPS system, governed by Theorem~\ref{th:BoundonXR}.
\begin{figure}[htbp]
  \begin{center}
    \includegraphics[scale=0.8]{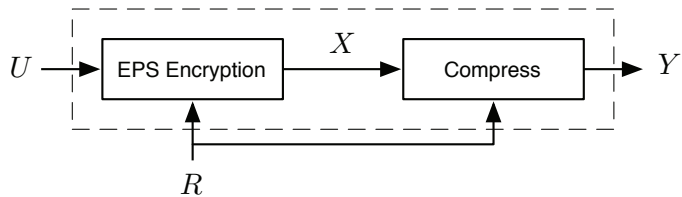} \caption{\label{fig:HXgivenR}
      Compressing the output of an EPS cipher.}
  \end{center}
\end{figure}

To complete this section, we show that the lower
bounds~\eqref{eq:thm1hx} and~\eqref{eq:thm1hr} are simultaneously
achievable using a one-time pad~\cite{OTP}.
\begin{definition}[One-time pad]\label{def:OTP}
  Without lost of generality, let $\calU = \{0, \ldots, M-1\}$ be the
  support of $U$.  Let $R$ be independent of $U$ and uniformly
  distributed in $\calU$ and let $X$ be generated according to the
  \emph{one-time pad} as $X = (U + R) \mod M$.  Then $U$ can be
  recovered via $(X + R) \mod M$.
\end{definition}

It is easy to verify that \eqref{eq:IUX0} -- \eqref{eq:IUR0} are satisfied and $H(X) = H(R) = \log M$.
Therefore, we have proved the following theorem.
\begin{theorem}[Achieving the minimum $H(X)$ and $H(R)$]
  \label{th:su}
  Let $\calU$ be the support of $U$.  The one-time pad of
  Definition~\ref{def:OTP} is an EPS system achieving $H(X) = \log
  |\calU|$ and $H(R) = \log |\calU|$.
\end{theorem}

\section{Multiple Messages and Key Consumption}\label{se:MultUse}
In Section~\ref{se:SingleUse}, Theorem~\ref{th:su} proved that the
one-time pad is ``optimal'' in the sense that it simultaneously
minimizes $H(X)$ and $H(R)$. This immediately suggests that the
one-time pad leaves no room for improvement. However, this conclusion
in fact stems from a folk theorem that the ``required size of the
secret key'' is measured by the key entropy. The hidden assumption
behind this folklore is that the \emph{cipher system is used only
  once}. In typical practice, a cipher system will be used repeatedly
for the transmission of multiple messages.


Consider the following scenario.  Suppose an initial secret key $R$ is
delivered to the sender and the receiver prior to commencement of
message transmission.  Now, suppose the sender uses this key to
encrypt a message $U$, which is then delivered to the receiver over
the public channel. Clearly, some portion of the secret randomness $R$
has now been used.  The central question is as follows: Can the sender
and receiver continue to securely communicate without first receiving
a new key?  For example, if $U$ is a single bit and $R$ is a 100-bit
random key, it is indeed likely that another message can be securely
transmitted. The natural questions are: What is the maximum size of
the second message?  Alternatively, how much of the key $R$ was
consumed in the first round of transmission?  Below, we will show that
when an error free perfect secrecy system is used multiple times, the
key consumption should not be measured by $H(R)$ but by $I(R; UX)$. In
fact, with respect to our definitions, we will exhibit systems with
key consumption that can be made arbitrarily close to $H(U)$.

The following example illustrates some of the basic ideas which will
be elaborated in this section.
\begin{example} \label{eg:KeyConsumption} Suppose the sender and the
  receiver share a secret key $R = \{B_1, B_2, \ldots, B_n\}$, where
  all of the $B_i$, $i=1,2,\dots,n$ are independent and uniformly
  distributed over $\{0,1\}$.  Let $P_U(0) = 0.5$ and $P_U(1) = P_U(2)
  = 0.25$.  Construct a new random variable $U'$ such that
  \begin{equation}
    U' =
      \begin{cases}
        \left(0, B_{n+1}\right), & U = 0 \\
        (1, 0), & U = 1 \\
        (1, 1), & U = 2
      \end{cases}
     \label{eq:defUpi}
  \end{equation}
  where $B_{n+1}$ is generated by the sender independently of $U$ and
  $R$ such that $P_{B_{n+1}}(0) = P_{B_{n+1}}(1) = 0.5$.

  Let $K = (B_1, B_2)$ and $X = U' \oplus K$. Upon receiving $X$, the
  receiver can decode $U$ from $X$ and $K$, where $K$ is solely a
  function of $R$. In fact, if $U=0$, the receiver can further decode
  $B_{n+1}$.  Let
  \begin{equation*}
    R' =
      \begin{cases}
        (B_3, B_4, \ldots, B_n), & U \in\{1,2\} \\
        (B_3, B_4, \ldots, B_n, B_{n+1}), & U = 0.
      \end{cases}
  \end{equation*}
  We refer to $R'$ as the \emph{residual secret randomness} shared by
  the sender and the receiver.  Note that $R'$ may not be a
  deterministic function of $R$, as the new shared common randomness
  can be generated by a probabilistic encoder. According
  to~\eqref{eq:defUpi}, a new random bit is secretly transmitted from
  the sender to the receiver when $U=0$.  After the system is used
  once, the expected key consumption is therefore given by
  \begin{equation}
    P_U(0) \cdot 1 + P_U(1) \cdot 2 + P_U(2) \cdot 2 = 1.5 =
    H(U), \label{eq:expectK}
  \end{equation}
  which happens to also equal $I(R;UX)$. It turns out that this is not
  mere coincidence.
\end{example}

We now define three parameters whose operational meanings are
justified in the rest of this section.
\begin{definition} \label{de:HRUX} The \emph{residual secret
    randomness} of an error free perfect secrecy system is
  \begin{equation*}
    H\left(R\given UX\right).
  \end{equation*}
\end{definition}

\begin{definition} \label{de:IRUX} The \emph{expected key consumption}
  of an error free perfect secrecy system is
\begin{equation*}
  I(R; UX).
\end{equation*}
\end{definition}

\begin{definition} \label{de:IRX} The \emph{excess key consumption} of
  an error free perfect secrecy system is\
  \begin{equation*}
    I(R; X).
  \end{equation*}
\end{definition}

Roughly speaking, we will show that after an EPS system is used once,
$H(R\given UX)$ is the amount of remaining key that can be used for
encryption of the next message.  Since the sender and the receiver
initially share a quantity $H(R)$ of secret randomness, the key
consumption is equal to $H(R) - H(R|UX) = I(R; UX)$.  We will provide
achievable schemes to show that the minimal key consumption is $H(U)$
and hence, the excess key consumption is $I(R; UX) - H(U)$ which is
equal to $I(R; X)$ in an EPS system.

We first justify Definition~\ref{de:HRUX}. Consider the scenario of
Fig.~\ref{fig:UVXY} in which the sender and receiver share a secret
key $R$, and two EPS systems are used sequentially by the sender to
securely transmit two (possibly correlated) messages $U$ and $V$.  In
the first round, the sender encodes the message $U$ into $X$, which is
transmitted to the receiver as described in
Section~\ref{se:SingleUse}. In the second round, the sender further
encodes $V$ (or more generally both $U$ and $V$) into $Y$, which is
then transmitted to the receiver.  As before, we require $H(U \given
R X) = H(V\given R X Y) = 0$ and $I(UV ; XY)=0$ to ensure zero-error
decoding and perfect secrecy.
\begin{figure}[htbp]
  \begin{center}
    \includegraphics[scale=0.8]{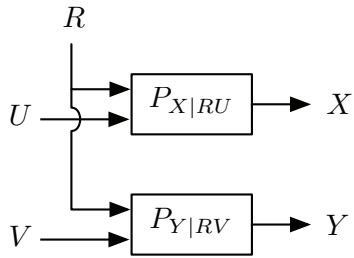}
    \caption{\label{fig:UVXY}Using an error free prefect secrecy
      system twice.  }
  \end{center}
\end{figure}

\begin{theorem}[Justification 1] \label{th:NextV} Consider the
  two-round error free perfect secrecy system of Fig.~\ref{fig:UVXY}.
  If
  \begin{equation}
    I(U V; X Y) = H(U \given R X) = H(V \given R X Y) = 0, \label{eq:NextV0}
  \end{equation}
  then the entropy of the second message $V$ conditioning on the first
  message $U$ is upper bounded by the residual secret randomness,
  \begin{eqnarray}
    H(V \given U) \leq H(R \given U X). \label{eq:NextV1}
  \end{eqnarray}
\end{theorem}
\begin{proof}
  Note that
  \begin{multline*}
    H(R \given  U  X) - H(V \given  U) + I(U V;X  Y) +
    H(U \given  R  X) +H(V \given  R  X   Y)  \\
    =  I(V  Y; U \given R  X) + I(R  U; Y\given X) + I(U; X) +
    H(U\given  R  X  Y) + H(R\given U  V  X  Y) \geq 0.
  \end{multline*}
  Together with~\eqref{eq:NextV0}, \eqref{eq:NextV1} is verified.
\end{proof}

Theorem~\ref{th:NextV} implies that the maximum amount of information
which can be secretly transmitted in the second round is upper bounded
by the residual secret randomness $H(R\given UX) $, suggesting that
$H(R\given UX) $ is indeed measures the amount of key unused in the
first round.  Equivalently, the amount of key that has been consumed
in the first round is equal to
\begin{equation*}
  I(R;UX) = H(R) - H(R\given UX).
\end{equation*}

Whereas Theorem~\ref{th:NextV} justifies the residual key $H(R\given
UX)$ as bounding the entropy of the second round message, we now offer
an alternative justification, showing that the size of the key that
can be extracted after $n$ uses of an EPS system is about $n H(R\given
UX)$.

Consider generation of a new secret key as shown in
Fig.~\ref{fig:2ndR}.  Suppose a sequence of EPS systems $\{(U_i, R_i,
X_i)\}_{i =1}^n$ has been used by a sender and a receiver where $(U_i,
R_i, X_i)$ are i.i.d. with generic distribution $P_{U R X}$.  We use
$(U,R,X)$ to denote the generic random variables.  In order to
securely send additional messages, the sender and the receiver aim to
establish a new secret key $S^m = (S_1, \ldots, S_m)$, where the $S_i$
are i.i.d. with generic distribution $P_S$. To generate the new key
$S^{m}$, we assume that the sender can send a secret message $A$ to
the receiver.  The new secret key $S^m$ will be used to encrypt a
second sequence of messages $V^m$, generating a ciphertext sequence
$Y^m$ such that $\{(V_i, S_i, Y_i)\}_{ i =1}^m$ is another sequence of
EPS systems.
\begin{figure}[htbp]
  \begin{center}
    \includegraphics[scale=0.8]{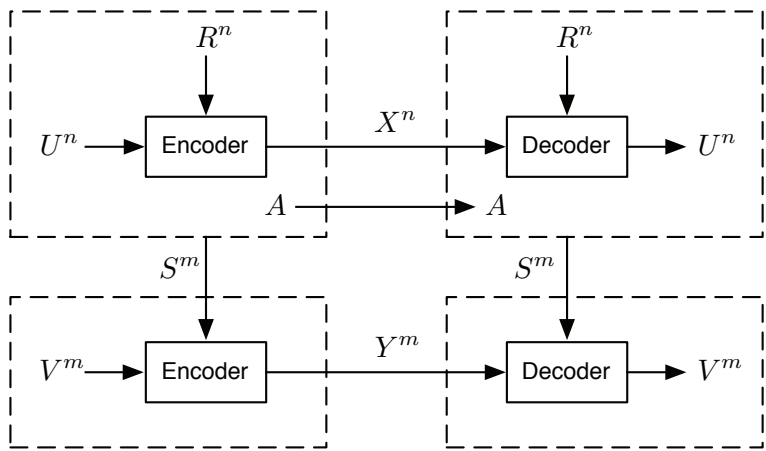}
    \caption{\label{fig:2ndR}Generating a new secret key $S^m$.
    }
  \end{center}
\end{figure}

Assume
\begin{align}
  I\left(V^m; S^m  U^n  X^n\right) &= 0 \label{eq:NextSreq28}\\
  I\left(U^n  X^n; Y^m \given  V^m  S^m\right) &= 0 \label{eq:NextSreq29}\\
  H\left(S^m \given  R^n  U^n  X^n  A\right) &= 0  \label{eq:NextSreq3}\\
  I\left(S^m; U^n  X^n\right) &= 0. \label{eq:NextSreq4}
\end{align}
These assumptions adopted with the following reasoning.  We assume in
\eqref{eq:NextSreq28} that the new message $V^m$ is generated
independently of the previous uses of the EPS systems.
Also,~\eqref{eq:NextSreq29} holds due to $(U^n,  X^n) - (V^m,  S^m) - Y^m $ forms a Markov chain. The sender and the receiver can agree on $S^m$
without error due to~\eqref{eq:NextSreq3}.  The justification
of~\eqref{eq:NextSreq4} is given as follows.

Although $\{(U_i, R_i, X_i)\}_ {i =1}^n$ and $\{(V_i, S_i, Y_i)\}_{ i
  =1}^m$ are individually sequences of EPS systems, it is possible
that their combination is not secure, $I(X^n, Y^m ; U^n, V^m) > 0$.
For example, suppose $U^n$ and $V^m$ are i.i.d. with uniform
distribution and $m=n$.  If $S^m = U^n$, then using a one-time pad,
$I(V^m; Y^m) = 0$ but $I(X^n, Y^m ; U^n, V^m) \geq H(U^n)$.  The
following theorem shows that joint EPS systems satisfying
\eqref{eq:NextSreq28} -- \eqref{eq:NextSreq4} are still perfectly
secure.
\begin{theorem}
  \label{th:2ndeps}
  Consider two sequences of i.i.d. EPS systems $\{(U_i, R_i, X_i)\}_
  {i =1}^n$ and $\{(V_i, S_i, Y_i)\}_{ i =1}^m$ satisfying
  \eqref{eq:NextSreq28} -- \eqref{eq:NextSreq4}.  Then the joint EPS
  system is still perfectly secure,
  \begin{equation}
    I\left(X^n Y^m ; U^n V^m\right) = 0. \label{eq:2ndeps1}
  \end{equation}
\end{theorem}

\begin{proof}
  By assumption
  \begin{multline}
     I(U^n; X^n) = I(V^m; Y^m) = I(V^m; S^m, U^n, X^n) = \\
     I(U^n, X^n; Y^m | V^m, S^m) = I(S^m; U^n, X^n) = 0. \label{eq:2ndeps0}
  \end{multline}
  Note that
  \begin{multline*}
    I\left(S^m; U^n,  X^n\right) + I\left(V^m; S^m,  U^n,  X^n\right) +
    I\left(V^m; Y^m\right) + \\
     I\left(U^n; X^n\right) + I\left(U^n,  X^n; Y^m \given  V^m  S^m\right) -
      I\left(X^n, Y^m ; U^n,  V^m\right) \\
    = I\left(U^n, X^n; S^m\given V^m, Y^m\right)  + I\left(V^m;S^m\right)
    + I\left(X^n;Y^m\right) +  I\left(U^n;V^m\right)
    \geq 0.
  \end{multline*}
  Together with~\eqref{eq:2ndeps0}, $I(X^n, Y^m ; U^n, V^m) \leq
  0$. Since $I(X^n, Y^m ; U^n, V^m) \geq 0$,~\eqref{eq:2ndeps1} is
  verified.
\end{proof}

In order to generate a new key $S^m$, a secret auxiliary random
variable $A$ is sent from the sender to the receiver.  Here, $A$ is
generated by a probabilistic encoder with $\{(R_i, U_i, X_i)\}_ {i
  =1}^n$ as input.  In Example~\ref{eg:KeyConsumption} above, suppose
we wanted to restore $n$ secret bits after the system is used
once. Then $A$ is a fair bit if $U = 0$ and $A$ consists of two fair
bits if $U = 1$ or $2$.  We measure the expected size of $A$ by
$H\left(A\given R^n U^n X^n\right)$.  Since we can directly treat $A$
as the new secret key $S^m$, it is reasonable to expect that $H(S^m)
\geq H\left(A\given R^n U^n X^n\right)$.  Therefore, it is of interest
to know by how much $H(S^m)$ can exceed $H(A\given R^n, U^n, X^n)$ for a
given sequence of EPS systems.  The following theorem shows that the
secret randomness, which can be extracted from $\{(R_i, U_i, X_i)\}_{i
  =1}^n$ with help from $A$, is measured by the residual secret
randomness $H(R\mid U, X)$.
\begin{theorem}[Justification 2]
  \label{th:NewKey}
  Consider two sequences of i.i.d. EPS systems $\{(U_i, R_i, X_i)\}_
  {i =1}^n$ and $\{(V_i, S_i, Y_i)\}_{ i =1}^m$ and any $A$.  If
  \eqref{eq:NextSreq28} -- \eqref{eq:NextSreq4} are satisfied, then
  \begin{equation}
    H\left(S^m\right) - H\left(A\given R^n U^n X^n\right) \leq  n
    H\left(R\given U X\right). \label{eq:NewKey2}
  \end{equation}
  On the other hand, it is possible to generate $S^m$ such that
  \eqref{eq:NextSreq28} -- \eqref{eq:NextSreq4} are satisfied and
  \begin{equation}
    H\left(S^m\right) - H\left(A\given R^n U^n X^n\right) \geq n
    H\left(R\given U X\right) - \log 2 \label{eq:NewKey3}
  \end{equation}
  for a sufficiently large $m$ such that $$\max_{s^m} P_{S^m}(s^m) <
  \min_{r^n, u^n ,x^n} P_{R^n\given U^n X^n}\left(r^n\given u^n, x^n\right).$$
\end{theorem}

\begin{proof}
  We first prove \eqref{eq:NewKey2} by showing that
  \begin{align}
    H\left(S^m\right)  &= I\left(S^m; U^n  X^n\right) +
    H\left(S^m\given  A  R^n  U^n  X^n\right) +
      I\left(S^m; A  R^n \given  U^n  X^n\right)\\
    &= I\left(S^m; A  R^n \given  U^n  X^n\right)  \label{eq:isar1}\\
    &\leq H\left(A  R^n \given  U^n  X^n\right)  \\
    &= H\left(A \given  R^n  U^n  X^n\right) + H\left(R^n \given  U^n
      X^n\right)\\
    &= H\left(A \given  R^n  U^n  X^n\right) + n\,H\left(R \given  U
      X\right),  \label{eq:isar2}
  \end{align}
  where~\eqref{eq:isar1} follows from \eqref{eq:NextSreq3} --
  \eqref{eq:NextSreq4} and~\eqref{eq:isar2} follows from the fact
  that $\{(U_i, R_i, X_i)\}_ {i =1}^n$ is a sequence of i.i.d. EPS
  systems.

  The proof of the achievability part in~\eqref{eq:NewKey3} is via
  construction.  With reference to Fig.~\ref{fig:defineA}, consider
  two partitions of the unit interval into disjoint ``cells''.  The
  width of cell $i$ in the first partition is $P_{S^m}(i)$ for $1 \leq
  i \leq |\calS|^m$, where $\calS$ is the support of $S_i$.  Consider
  $U^n = u^n$ and $X^n = x^n$.  The width of cell $i$ in the second
  partition is $P_{R^n\given U^n, X^n}\left(i\given u^n, x^n\right)$
  for $1 \leq i \leq |\calR|^n$.  The distribution of $A$ is
  constructed to divide the second partition as shown in
  Fig.~\ref{fig:defineA}.
  \begin{figure}[htbp]
    \begin{center}
      \includegraphics[scale=1]{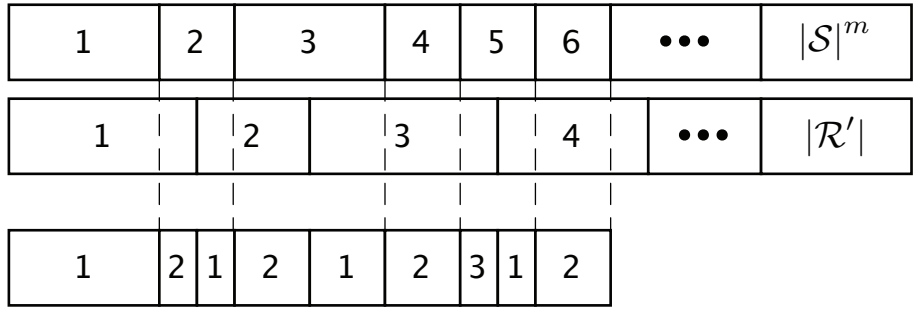}
      \caption{\label{fig:defineA}
        An assignment of $A$.}
    \end{center}
  \end{figure}
%
To simplify notations, we consider the support of $R^n$ to be a set of consecutive integers $\{1, ..., |\calR'|\}$ when $U^n = u^n$ and $X^n = x^n$.   Suppose $R^n = r$ and let
\begin{align}
    b = \max \left\{j: \sum_{i=1}^j P_S(i) > \sum_{i=1}^{r-1} P_{R^n\given U^n, X^n}\left(i\given u^n, x^n\right) \right\}.
\end{align}
For $j \geq 1$, $A$ is defined by $P_A(j) = \frac{a(j)}{P_{R^n\given U^n, X^n}\left(r\given u^n, x^n\right) }$, where
  \begin{align}
        \sum_{i=1}^{j} a(i)  &= \min \left\{\sum_{i=1}^{b + j-1} P_S(i),  \sum_{i=1}^{r} P_{R^n\given U^n, X^n}\left(i\given u^n, x^n\right) \right\} -  \sum_{i=1}^{r-1} P_{R^n\given U^n, X^n}\left(i\given u^n, x^n\right).
  \end{align}
For the example in Fig.~\ref{fig:defineA}, when $R^n = 1$,
  \begin{align}
    P_A(1) = \frac{P_{S^m}(1)}{P_{R^n\given U^n, X^n}\left(1\given u^n, x^n\right)} = 1 - P_A(2).
  \end{align}
  When $R^n = 2$,
  \begin{align}
    P_A(1)  = \frac{P_{S^m}(1) + P_{S^m}(2) - P_{R^n\given U^n, X^n}\left(1\given u^n, x^n\right)}{P_{R^n\given U^n, X^n}\left(2\given u^n, x^n\right)} =
    1 - P_A(2).
  \end{align}

  By definition $S^m$ is determined from $R^n$ and $A$ for any fixed
  $U^n = u^n$ and $X^n = x^n$.  On the other hand, $A$ is also
  determined from $S^m$ and $R^n$.  Therefore,
  \begin{align}
    H\left(S^m \given A, R^n, U^n, X^n\right) = H\left(A \given S^m,
      R^n, U^n, X^n\right) = 0. \label{eq:nk1}
  \end{align}
  By choosing $m$ sufficiently large, such that
  \begin{align}
    \max_{s^m} P_{S^m}(s^m) < \min_{r^n, u^n ,x^n} P_{R^n\given U^n,
      X^n}\left(r^n\given u^n, x^n\right),
  \end{align}
  $R^n$ can take at most two possible values for any given $\left(S^m, U^n,
  X^n\right)$ and hence
  \begin{align}
    H\left(R^n\given S^m, U^n, X^n\right) \leq \log 2. \label{eq:nk15}
  \end{align}
  Therefore,
  \begin{align}
    \lefteqn{H\left(A \given  R^n  U^n  X^n\right)} \\
    &= I\left(A ; S^m \given  R^n  U^n  X^n\right) + H\left(A \given
      S^m  R^n  U^n  X^n\right)\\
    &= I\left(A ; S^m \given  R^n  U^n  X^n\right) + H\left(S^m \given
      A  R^n  U^n  X^n\right)  \label{eq:nk2}\\
    &= H\left(S^m \given  R^n  U^n  X^n\right)  \label{eq:nk3}\\
    &= H\left(S^m \right) - H\left(R^n\given  U^n  X^n\right) +
    H\left(R^n\given  S^m  U^n  X^n\right) - I(S^m; U^n, X^n)\\
    &\leq H\left(S^m \right) - H\left(R^n\given  U^n  X^n\right) +
    H\left(R^n\given  S^m  U^n  X^n\right) \\
    &\leq H\left(S^m \right) - H\left(R^n\given  U^n  X^n\right) +
    \log 2,  \label{eq:nk4}
  \end{align}
  where~\eqref{eq:nk2} and~\eqref{eq:nk4} follow from~\eqref{eq:nk1}
  and~\eqref{eq:nk15}, respectively.  Since $\left\{\left(U_i, R_i,
      X_i\right)\right\}_ {i =1}^n$ is a sequence if i.i.d. EPS
  systems, \eqref{eq:NewKey3} is verified.

  For any $\left(U^n, X^n\right)$, the same $P_{S^m\given U^n, X^n}
  \equiv P_{S^m}$ is generated.  Therefore,~\eqref{eq:NextSreq4} is
  verified.  Since $S^m$ is determined by $\left(R^n, U^n, X^n,
    A\right)$,~\eqref{eq:NextSreq3} is verified,
  and~\eqref{eq:NextSreq28} can also be verified as $V^m$ is
  independent of $\left(R^n, U^n, X^n, A\right)$.
  Finally,~\eqref{eq:NextSreq29} is due to the fact that
  $\left\{\left(V_i, S_i, Y_i\right)\right\}_{ i =1}^m$ is a sequence
  of EPS systems.
\end{proof}

Roughly speaking, Theorem~\ref{th:NewKey} shows that for large $n$ and
$m$, the optimal algorithm with the help of $A$ can extract
approximately
\begin{equation*}
  n H\left(R\given U X\right)
\end{equation*}
bits of residual secret randomness from $\left\{(R_i, U_i,
  X_i)\right\}_{i=1}^{n}$.  In~\cite{2DIA}, we considered another
algorithm generating a new secret key with asymptotic rate
$H\left(R\given U X\right)$ without using an auxiliary secret random
variable.  As the sender and receiver initially share $n H(R)$ bits of
secret randomness, the expected key consumption for each use of the
EPS system is
\begin{equation*}
  H(R) - H(R\given U  X) = I(R ; U X),
\end{equation*}
the quantity proposed in Definition~\ref{de:IRUX}.  Next, we exhibit an
important property of $I(R;UX)$.
\begin{theorem}\label{th:LBonEKC}
  In an error free perfect secrecy system, the expected key
  consumption is lower bounded by the source entropy,
  \begin{eqnarray}
    I(R; UX) \geq H(U), \label{eq:LBonEKC}
  \end{eqnarray}
  where equality holds if and only if $I(R; X) = 0$.
\end{theorem}
\begin{proof}
  The information diagram for the random variables $U, X, R$ involved
  in an error free perfect secrecy system satisfying~\eqref{eq:IUX0}
  -- \eqref{eq:IUR0} is shown in Fig.~\ref{fig:IDa}.  It is easy to
  verify that
  \begin{eqnarray}
    I(X; R) = I(R; UX) - H(U). \label{eq:excessK}
  \end{eqnarray}
  Since $I(X; R) \geq 0$, Theorem~\ref{th:LBonEKC} is proved.
\end{proof}
\begin{figure}[htbp]
  \begin{center}
    \subfigure[General EPS system\label{fig:IDa}]{\includegraphics[scale=0.8]{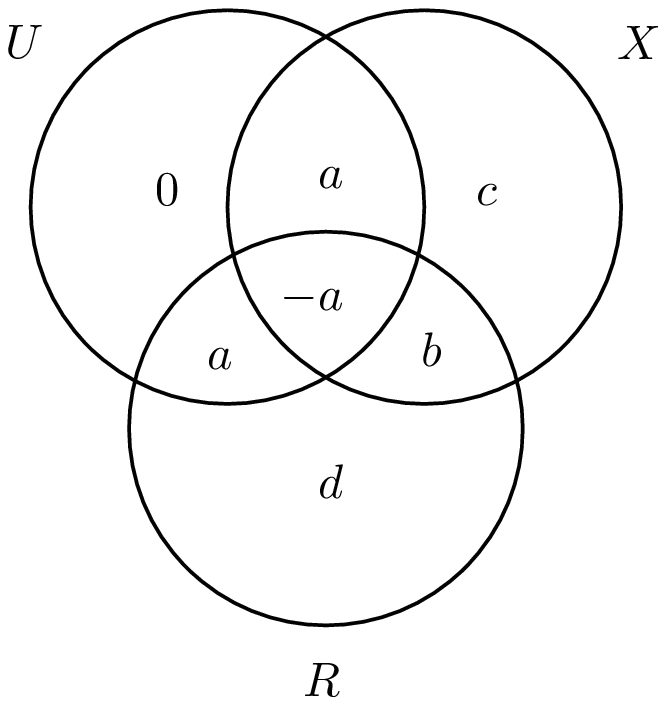}}
    \hspace{10mm}
    \subfigure[Minimum expected key consumption,
    achieving equality in~\eqref{eq:LBonEKC}\label{fig:IDb}]{\includegraphics[scale=0.8]{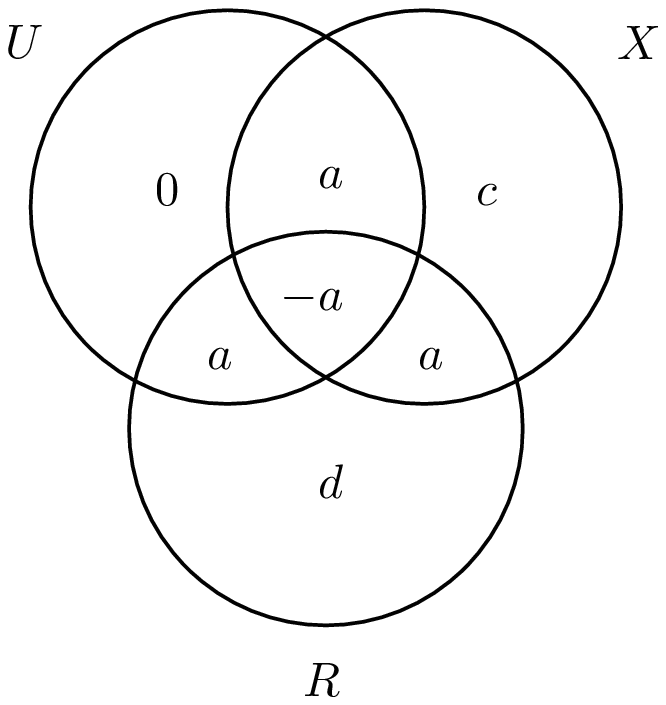}}
    \caption{\label{fig:ID}Information diagrams.}
  \end{center}
\end{figure}

In Section~\ref{se:sub1MinIRUX}, we will describe several EPS coding
schemes achieving $I(R; UX) = H(U)$.  Therefore, $I(X;R)$ measures the
difference between the expected key consumption of an EPS system and
the minimum possible key consumption, again justifying
Definition~\ref{de:IRX}.  The information diagram for the optimal case
$I(X;R)=0$ is shown in Fig.~\ref{fig:IDb}.

We summarize this section in the following three remarks.
\begin{enumerate}
\item Theorems~\ref{th:NextV} and~\ref{th:NewKey} provide strong
  justification of $I(R;UX)$ as the expected key consumption required
  to achieve error free perfect secrecy.  Theorem~\ref{th:LBonEKC}
  shows that the expected key consumption cannot be less than the
  source entropy.  Recall that Theorem~\ref{th:BoundonXR} gives the
  lower bound on the initial key requirement.  Therefore, we have
  distinguished between two different concepts (a) expected key
  consumption in a multi-round system and (b) the initial key
  requirement for a one-shot system. In contrast to the bound $H(R)
  \ge H(U)$~\cite{shannon1949communication,massey1992},
  Theorem~\ref{th:LBonEKC} more precisely describes the role of $H(U)$
  in an error free perfect secrecy system.

\item From \eqref{eq:IUX0}--\eqref{eq:IUR0} we can show that
  \begin{eqnarray}
    H(R) = H(U) + I(X; R) + H(R\given U X).
  \end{eqnarray}
  Thus the key entropy $H(R)$ consists of three parts: the randomness
  used to protect the source, the excess key consumption and the
  residual secret randomness.

\item If the source distribution is uniform, Example~\ref{eg:uniform}
  below shows that the one-time pad achieves minimal key consumption.
\end{enumerate}

\begin{example}[Uniform source distribution] \label{eg:uniform}
  Suppose $U$ and $R$ are independent and are uniformly distributed on
  the sets $\{0, 1, \ldots, 2^i-1\}$ and $\{0, 1, \ldots, 2^j-1\}$,
  respectively, where $i \leq j$.  In order to derive a coding system
  satisfying \eqref{eq:IUX0} -- \eqref{eq:IUR0}, we can first extract
  $i$ random bits $R'$ from $R$ and construct $X$ as the modulo-two
  addition of the binary representation of $U$ and $R'$.  Then
  \begin{align}
    I(R; UX) &= H(R) - H(R\given UX) \\
    &= H(R) - H(R\given R') \\
    &= j - (j - i) \\
    &= H(U).
  \end{align}
\end{example}

\section{Tradeoff between Key Consumption and Number of Channel Uses}
\label{se:Tradeoff}
Example~\ref{eg:uniform} shows that the one-time pad simultaneously
achieves the minimal expected key consumption and the minimum number
of channel uses for a uniform source. However for general non-uniform
sources, we will show that there is a non-trivial tradeoff between
these two quantities.

We will consider two important regimes. First, in
Section~\ref{se:sub1MinIRUX}, we will consider the regime in which
ciphers minimize the key consumption $I(R; UX)$. Conversely, in
Section~\ref{se:sub2MinHX} we consider systems which minimize the
number of channel uses $H(X)$.

We shall demonstrate the existence of a fundamental, non-trivial
tradeoff between the expected key consumption and the number of
channel uses.  Our main results, Theorem~\ref{th:BoundonXR} proved
earlier, and Theorems~\ref{th:LBonEKC} -- \ref{th:minHX} to be proved
below, are summarized in Fig.~\ref{fig:IRUXvsHX}.

Point $\mathsf{1}$ is due to Theorem~\ref{th:minHX} in
Section~\ref{se:sub2MinHX} below, and has the smallest $I(R;UX)$ among
all EPS systems with $H(X)=\log|\calU|$. We shall show that this point
can always be achieved by one-time pad.

Point $\mathsf{2}$ has the smallest $H(X)$ among all the EPS systems
with $I(R; UX)= H(U)$.  For this point, Theorem~\ref{th:IXR0Bound} in
Section~\ref{se:sub1MinIRUX}
 gives the lower bound on $H(X)$ which is strictly greater than $\log
|\calU|$ if $P_U$ is not uniform.

If $P_U$ has only rational probability masses,
Theorem~\ref{th:LBonEKCAch} in Section~\ref{se:sub1MinIRUX} below
shows that Point $\mathsf{3}$ can be achieved by a generalization of
the one-time pad, the \emph{partition code} (to be introduced in
Definition~\ref{def:partitioncode}).

If all the probabilities masses in $P_U$ are the integer multiples of
the smallest probability mass in $P_U$, then Point $\mathsf{2}$
coincides with Point $\mathsf{3}$ by the partition code shown in
Theorem~\ref{th:optimalPart}.  Otherwise, Point $\mathsf{3}$ can
differ from Point $\mathsf{2}$ which will be demonstrated in
Example~\ref{eg:notpartition}.

The existence, continuity and non-increasing in $H(X)$ properties of
the curved portion of the tradeoff curve are established in
Section~\ref{sec:tradeoff}.
\begin{figure}[htbp]
  \begin{center}
    \includegraphics[scale=1]{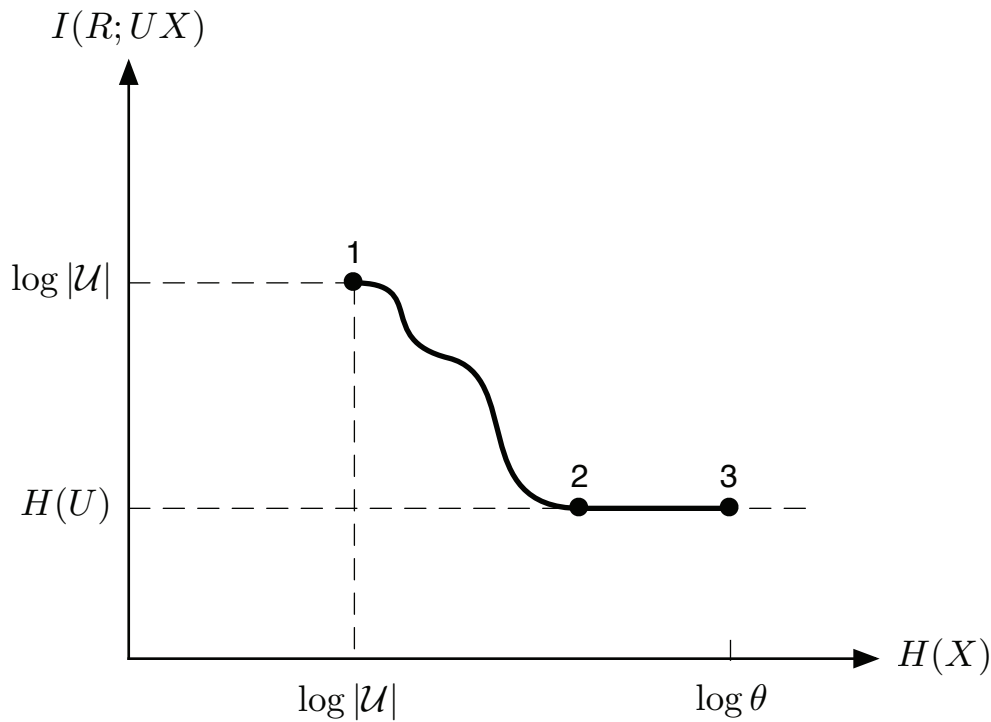} \caption{\label{fig:IRUXvsHX}
      Tradeoff between $I(R;UX)$ and $H(X)$.}
  \end{center}
\end{figure}

\subsection{Minimal expected key consumption}
\label{se:sub1MinIRUX}
We first consider EPS systems which achieve minimal expected key
consumption.  From Theorem~\ref{th:LBonEKC}, an error free perfect
secrecy system with minimal key consumption satisfies
\eqref{eq:IUX0}--\eqref{eq:IUR0} and
\begin{eqnarray}
  I(X; R) = 0. \label{eq:IXR0}
\end{eqnarray}
We now generalize the one-time pad to achieve minimal key consumption
for source distributions containing only rational probability masses.
\begin{definition}[Partition Code $\calC(\Psi)$]\label{def:partitioncode}
  Assume that $U$ is a random variable defined on $\{1,\ldots,
  \ell\}$.  Let $\Psi = (\psi_1, \psi_2, \ldots, \psi_\ell)$ and let
  $\theta = \sum_{i=1}^\ell \psi_i$ where $\psi_i$ and $\theta$ are
  positive integers.  Let $A'$ be a random variable such that
  \begin{equation*}
  \Pr(A' = j \given U=i ) =
  \begin{cases}
    \frac{1}{\psi_{i}} & \text{ if } 1\le j \le \psi_{i}, \\
    \:\: 0 & \text{ otherwise.}
  \end{cases}
  \end{equation*}
  Let $A = \sum_{i=1}^{U-1}\psi_i + A' - 1$, $R$ be uniformly
  distributed on the set $\{0, 1, \ldots, \theta -1 \}$ and $X = A + R
  \mod \theta$.  The so defined cipher system $(R,U,X)$ is called the
  \emph{partition code} $\calC(\Psi)$.
\end{definition}
Note that one-time pad is a special case of partition code when $\Psi
= (1, 1, \ldots, 1)$.

It can be proved directly that a partition code satisfies
\eqref{eq:IUX0} -- \eqref{eq:IUR0} and hence is an EPS system.
Furthermore, we can verify that
\begin{align}
  H(X) &= H(R) = \log \theta, \label{eq:PcodeHXR}\\
  \intertext{and}
  I(R; UX) &= \sum_{i=1}^\ell P_U(i) \log
  \frac{\theta}{\psi_i}, \label{eq:PcodeIRUX}
\end{align}
where~\eqref{eq:PcodeIRUX} is from $H(X\given U, R) = H(A\given U, R)
= \sum_{i=1}^\ell P_U(i) \log {\psi_i}$ together
with~\eqref{eq:PcodeHXR}.

Let $Q_U$ be the probability distribution such that $Q_U(i) =
\psi_i/\theta$.  Then \eqref{eq:PcodeIRUX} can be rewritten as
\begin{eqnarray}
  I(R; UX)  = H(U) + D(P_U \Vert Q_U),  \label{eq:IRUXD}
\end{eqnarray}
where $D(\cdot\Vert \cdot)$ is the relative entropy~\cite{bk:Cover}.
Consequently, we have the following theorem.

\begin{theorem} \label{th:LBonEKCAch} Suppose the probability mass
  $P_U(i)$ is rational for all $i = 1, \ldots, \ell $.  Let $\theta$
  be an integer such that $\theta \cdot P_U(i)$ is also an integer for
  all $i$, and let $\Psi = (\psi_1, \psi_2, \ldots, \psi_\ell)$ with
  $\psi_i = \theta \cdot P_U(i)$.  Then the EPS system $(R,U,X)$
  induced by the partition code $\calC(\Psi)$ achieves the lower bound
  in \eqref{eq:LBonEKC}, namely  $I(R; UX) = H(U)$.
\end{theorem}

%
%
%
%


In the following theorem, we prove that if the source distribution
$P_{U}$ is not rational, then partition code will not achieve zero
key-excess with finite $X$ or $R$. Its proof is deferred to
Section~\ref{se:supportXR}.

\begin{theorem}\label{th:supportXR}
  Suppose $U$, $X$, and $R$ satisfy~\eqref{eq:IUX0} -- \eqref{eq:IUR0}
  and~\eqref{eq:IXR0}.  If there exists $u\in\calU$ such that $P_U(u)$
  is irrational, then the support of $X$ and $R$ cannot be finite.
\end{theorem}

Although it is difficult to construct codes satisfying~\eqref{eq:IUX0} -- \eqref{eq:IUR0} and~\eqref{eq:IXR0} for $P_U$ having irrational
probability masses, Theorem~\ref{th:LBonEKC} still gives a tight bound
on $I(R; UX)$ as shown in the following theorem.
\begin{theorem} \label{co:irrational3} Suppose the support of $P_{U}$
  is a finite set of integers $\{1, \ldots, \ell\}$. Let $\Psi =
  (\psi_1, \ldots, \psi_{\ell+1})$ with
  \begin{equation*} \psi_i=
    \begin{cases}  \lfloor P_U(i)\theta \rfloor, & 1\leq i
      \leq \ell,\\
      \theta-\sum_{i=1}^{\ell}\lfloor P_U(i)\theta \rfloor, &   i=\ell+1.\\
    \end{cases}
  \end{equation*}
  Assume that $\theta$ is large enough such that $\lfloor P_U(i)\theta
  \rfloor \geq 1$ for all $1\le i \le \ell$.  For the partition code
  $\calC(\Psi)$, $I(R; UX) \rightarrow H(U)$ as $\theta \rightarrow
  \infty$.
\end{theorem}

\begin{proof}
  Consider a probability distribution $Q_U$ with $Q_U(i) =
  {\psi_i}/{\theta}$ for $1 \leq i \leq \ell+ 1$.  As $\theta
  \rightarrow \infty$, $Q_U$ converges pointwise to $P_U$ and hence
  $D(P_U \Vert Q_U) \rightarrow 0$ for finite $\ell$.  The theorem
  thus follows from~\eqref{eq:IRUXD}.
\end{proof}

In addition to minimizing the key consumption $I(R;UX)$, we may also
want to simultaneously minimize $H(X)$, which is the number of channel
uses required to convey the ciphertext $X$. The following theorem and
corollary illustrate that the zero key-excess condition can be very
harsh, requiring the EPS system to have a very large $H(R)$ and $H(X)$,
even for very simple sources.

\begin{theorem}[EPS systems with minimal $I(R;UX)$]
  \label{th:IXR0Bound}
  Let $\calX$, $\calR$ and $\calU$ be the respective supports of
  random variables $X$, $R$, and $U$ satisfying
  \eqref{eq:IUX0}--\eqref{eq:IUR0} and \eqref{eq:IXR0}. Then
  \begin{align}
    \max_{x\in\calX}P_X(x) &\leq \min_{u\in\calU } P_U(u) \label{eq:minIXR0PX}\\
  \intertext{and}
    \max_{r\in\calR} P_R(r) &\leq \min_{u\in\calU} P_U(u). \label{eq:minIXR0PR}
  \end{align}
\end{theorem}
\begin{proof}
  Consider any $u\in\calU$ and $x\in\calX$. By definition, $P_{U}(u) >
  0$ and $P_X(x) > 0$. From~\eqref{eq:IUX0}, we have $P_{UX}(u, x) =
  P_U(u) P_{X}(x) > 0$.  Consequently, there exists $r\in\calR$ such
  that $P_{UXR}(u, x, r) > 0$.  Notice that
  \begin{align}
    P_{UXR}(u, x, r)  &= P_{XR}(x, r)   \label{eq:thmx29}\\
    &= P_{X}(x) P_{R}(r), \label{eq:thmx299}
  \end{align}
  where~\eqref{eq:thmx29} is due to~\eqref{eq:HURX0}
  and~\eqref{eq:thmx299} is due to~\eqref{eq:IXR0}. On the other hand,
  \begin{align}
    P_{UXR}(u, x, r) &\leq P_{UR}(u, r) \label{eq:thmx2}\\
    &= P_{U}(u) P_{R}(r), \label{eq:thmx3}
  \end{align}
  where~\eqref{eq:thmx3} is due to~\eqref{eq:IUR0}. Finally, as
  $P_{R}(r)>0$, we have $P_{X}(x) \le P_{U}(u)$
  and~\eqref{eq:minIXR0PX} follows.  Due to the symmetric roles of $X$
  and $R$, the theorem is proved.
\end{proof}

The results in Theorem~\ref{th:IXR0Bound} are used to obtain bounds on
$H(X)$ and $H(R)$ in the following corollary.  Define the binary
entropy function, $h(\gamma) = -\gamma \log {\gamma} - (1-\gamma) \log
(1-\gamma)$ for $0 <\gamma < 1$ and $h(0) = h(1) = 0$.

\begin{corollary}
  \label{co:IXR0Bound}
  Let $\calX$, $\calR$ and $\calU$ be the respective supports of
  random variables $X$, $R$, and $U$ satisfying~\eqref{eq:IUX0} --
  \eqref{eq:IUR0} and~\eqref{eq:IXR0}. Then
  \begin{align}
    \min \{H(X), H(R) \} \geq& h ( \pi \lfloor \pi^{-1} \rfloor ) +
    \pi \lfloor \pi^{-1} \rfloor \log \lfloor \pi^{-1}
    \rfloor \label{eq:minIXR0HXHR} \\
    \geq& \log \frac{1}{\pi} \label{eq:minIXR0HXHR2},
  \end{align}
  where  $\pi = \min_{u\in\calU } P_U(u)$ and the right sides of
  \eqref{eq:minIXR0HXHR} and \eqref{eq:minIXR0HXHR2} are equal if and
  only if $\pi^{-1} $ is an integer.
\end{corollary}
\begin{proof}
  From~\eqref{eq:minIXR0PX}, $\max_{x\in\calX}P_X(x) \leq
  \min_{u\in\calU } P_U(u)$.  Together with~\cite[Theorem~10]{FanoJ},
  this establishes~\eqref{eq:minIXR0HXHR}.  To
  prove~\eqref{eq:minIXR0HXHR2}, we first consider the case when
  $\pi^{-1}$ is an integer.  Then
  \begin{align*}
    h( \pi \lfloor \pi^{-1} \rfloor ) + \pi \lfloor \pi^{-1} \rfloor
    \log \lfloor \pi^{-1} \rfloor
    =& h(1) + \pi \pi^{-1} \log \frac{1}{\pi}\\
    =& \log \frac{1}{\pi}.
  \end{align*}
  If $\pi^{-1}$ is not an integer, then
  \begin{equation*}
    1 - \pi \lfloor \pi^{-1} \rfloor < \pi.
  \end{equation*}
  Hence,
  \begin{eqnarray*}
    \lefteqn{h( \pi \lfloor \pi^{-1} \rfloor ) + \pi \lfloor \pi^{-1}
      \rfloor \log \lfloor \pi^{-1} \rfloor} \\
    &=& \pi \lfloor \pi^{-1} \rfloor \log \frac{1}{\pi \lfloor \pi^{-1}
      \rfloor} + (1- \pi \lfloor \pi^{-1} \rfloor) \log \frac{1}{1-\pi \lfloor
      \pi^{-1} \rfloor} +
    \pi \lfloor \pi^{-1} \rfloor \log \lfloor \pi^{-1} \rfloor\\
    &>& \pi \lfloor \pi^{-1} \rfloor \log \frac{1}{\pi }
    + (1- \pi \lfloor \pi^{-1} \rfloor) \log \frac{1}{\pi} \\
    &=& \log \frac{1}{\pi}.
  \end{eqnarray*}
  Furthermore, the right hand sides of~\eqref{eq:minIXR0HXHR}
  and~\eqref{eq:minIXR0HXHR2} are equal only if $\pi^{-1}$ is an
  integer.  This proves the lower bounds on $H(X)$.  Due to the
  symmetric roles of $X$ and $R$, the theorem is proved.
\end{proof}

Suppose $P_U$ is not uniform so that $\min_{u\in\calU } P_U(u) <
|\calU|^{-1}$.  In this case, \eqref{eq:minIXR0HXHR2} shows that
\begin{eqnarray}
  \min\{H(X), H(R)\}  > \log |\calU|.
\end{eqnarray}
Comparing with~\eqref{eq:thm1hx} and~\eqref{eq:thm1hr} in
Theorem~\ref{th:BoundonXR}, a larger initial key requirement and a
larger number of channel uses are required for systems which achieve
the minimal expected key consumption.  The following theorem shows
that the lower bounds in~\eqref{eq:minIXR0HXHR2} can be achieved for
certain $P_U$ including the uniform distribution and $D$-adic
distributions, $P_U(u) = D^{-i}$ for certain integers $D$ and~$i$.
\begin{theorem}
  \label{th:optimalPart}
  Let $\calU =\{1, \ldots, \ell\}$ and let $P_U(\ell) \leq P_U(i)$ for
  $1 \leq i \leq \ell$.  If there exists a set of positive integers
  $\Psi = \{\psi_i\}$ such that $P_U(i) = \psi_i P_U(\ell)$ for $1
  \leq i \leq \ell$, then the partition code $\calC(\Psi)$
  simultaneously achieves the minimum $H(X)$ and $H(R)$ among all EPS
  systems achieving minimal key consumption.
\end{theorem}

\begin{proof}
  Suppose $(R, U, X)$ satisfies~\eqref{eq:IUX0} -- \eqref{eq:IUR0}
  and~\eqref{eq:IXR0} so that $H(X) \geq \log \frac{1}{P_U(\ell)}$
  from~\eqref{eq:minIXR0HXHR2}.  Note that $P_U(\ell) =
  (\sum_{i=1}^\ell \psi_i)^{-1}$ from the definition of $\Psi$.
  Therefore
  \begin{equation}
    H(X) \geq \log \left(\sum_{i=1}^\ell \psi_i \right). \label{eq:lBonHX}
  \end{equation}
  The partition code $\calC(\Psi)$ has $\theta = \sum_{i=1}^\ell \psi_i$
  so that it can achieve equality in \eqref{eq:lBonHX} from
  \eqref{eq:PcodeHXR}.  Similarly, we can argue that the partition
  code $\calC(\Psi)$ achieves the minimum $H(R)$.
\end{proof}

For some other source distributions $P_U$, the partition code may not
achieve the minimal number of channel uses $H(X)$, as illustrated in
the following example.
\begin{example}
  \label{eg:notpartition}
  Consider an EPS system $(R,U,X)$ such that
  \begin{enumerate}
  \item $U$ is a binary random variables where  $P_{U}(0) = 3/5$.
  \item $X$ and $R$ take values from the set $\{0,1,2,3\}$.
  \item $P_{X}(0) =P_{R}(0) = 2/5$, $P_{X}(i) =P_{R}(i) = 1/5$
    for $i=1,2,3$.
  \item $I(X;R) = 0$ so that $P_{XR}(xr) = P_X(x) P_R(r)$ for all $x$ and $r$.
  \item $U$ is a function of $(X,R)$ such that $U=0$ if and only if
    (i) $X = 0$ and $R \neq 0$, or (ii) $R=0$ and $X\neq 0$, or (iii)
    $X=R\neq 0$. Consequently, $P_{U|XR}(u\given x,r)$ is well-defined.
  \end{enumerate}
  It is straightforward to check that $\{U,X,R\}$ satisfies
  \eqref{eq:IUX0} -- \eqref{eq:IUR0} and~\eqref{eq:IXR0} and $H(X) =
  H(R) < \log 5$.  However $\theta = 5$ is the smallest integer such
  that $\theta \cdot P_U(u)$ is an integer.  In this example, $H(X)$
  is smaller than the value given in \eqref{eq:PcodeHXR}.  While
  Theorem~\ref{th:optimalPart} shows that partition code can
  simultaneously minimize $H(X)$ and $H(R)$ under the conditions
  \eqref{eq:IUX0} -- \eqref{eq:IUR0} and~\eqref{eq:IXR0}, this
  example shows that partition code is not necessarily optimal in
  terms of minimizing $H(X)$ for a general source.
\end{example}

\subsection{Minimal number of channel uses}
\label{se:sub2MinHX}

In the previous subsection, we proposed partition codes $\calC(\Psi)$
which minimize the expected key consumption for error free perfect
secrecy systems. However, we also demonstrated that these codes do not
guarantee the minimal number of channel uses $H(X)$, among all other
EPS systems which also minimize the expected key consumption. Finding
an EPS system which minimizes the number of channel uses for a given
expected key consumption is a very challenging open problem. In this
subsection, we aim to minimize $I(R;UX)$ in the regime where $H(X)$
meets the lower bound in Theorem~\ref{th:BoundonXR}, $H(X) = \log
|\calU|$. Unlike in Section~\ref{se:sub1MinIRUX}, we can completely
characterize this regime.

Using Theorem~\ref{th:su}, we can show that by using one-time pad,
\begin{equation*}
  H(U) \le \log |\calU| = H(X) = H(R) = I(R; UX).
\end{equation*}
Therefore, in this instance, the expected key consumption $I(R; UX)$
is not minimal when the source $U$ is not uniform.  However, the
following theorem shows that among all EPS systems which minimize the
number of channel uses, the one-time pad minimizes the expected key
consumption.
\begin{theorem}
  \label{th:minHX}
  Consider any EPS system $(R,U,X)$ (e.g., one-time pad) with $H(X) =
  \log |\calU|$. Then $I(R; UX) = \log |\calU|$ and $H(X|RU) = 0$.
\end{theorem}
\begin{proof}
  If $H(X) = \log |\calU|$, $P_X(x) = {1}/{|\calU|}$ for $x \in \calX$ and
  \begin{equation}
    |\calX| = |\calU|   \label{eq:calXM}
  \end{equation}
  from Theorem~\ref{th:BoundonXR}.  Let
  \begin{equation*}
    \calX_{ru} = \{x \in \calX: P_{RUX}(r, u, x) > 0 \}
  \end{equation*}
  be the set of possible values of $X$ when $R = r$ and $U = u$.
  Due to \eqref{eq:cov1}, $\calX_{ri} \cap \calX_{rj} = \emptyset$ if
  $i \neq j$.  Together with \eqref{eq:calXM},
  \begin{equation}
    |\calU| = |\calX| \geq \left|\bigcup_u \calX_{ru} \right| = \sum_u
    |\calX_{ru}| \geq |\calU|\min_u |\calX_{ru}|. \label{eq:Sru1}
  \end{equation}
  On the other hand, for any $r \in \calR$ and $u \in \calU$
  \begin{equation*}
    \sum_{x \in \calX_{ru}} P_{RUX}(r, u, x) = P_{UR}(u,r) = P_{U}(u)P_R(r) > 0
  \end{equation*}
  from~\eqref{eq:IUR0}, and hence, $|\calX_{ru}| \geq 1$.
  Substituting this result into~\eqref{eq:Sru1} shows that $|\calX_{ru}| =
  1$.  Therefore, $X$ is a function of $R$ and $U$, which verifies
  \begin{eqnarray}
    H(X\given UR) = 0.
  \end{eqnarray}
  Together with \eqref{eq:IUX0} -- \eqref{eq:IUR0}, it is easy to
  verify that $I(R; UX) = H(X) = \log |\calU|$.
\end{proof}


\subsection{The fundamental tradeoff}\label{sec:tradeoff}
An important open problem is to find coding schemes which can achieve
points on the tradeoff curve between Points $\mathsf{1}$ and
$\mathsf{2}$ in Figure~\ref{fig:IRUXvsHX}.  For a given source
distribution $P_U$ and number of channel uses
$H(X)=\log|\calU|+\gamma$, with $\gamma\geq 0$ we need to solve the
following optimization problem,
\begin{equation}
  f(\gamma) = \inf_{P_{RX|U} \in \calP_\gamma} I(R; UX), \label{eq:minIRUX1}
\end{equation}
where
\begin{equation}\label{eq:feasibleset}
  \calP_\gamma = \left\{P_{RX|U}: I(R; U) = I(X; U)  = H(U|X, R) = 0,
    H(X) = \log |\calU| + \gamma \right\}
\end{equation}
is the set of feasible conditional distributions yielding an EPS
system with the specified number of channel uses.

Solving~\eqref{eq:minIRUX1} remains open in general, however two
important structural properties of $f(\gamma)$ are given in the
following theorem.
\begin{proposition}
  Let $P_U$ and $\gamma\geq 0$ be given. Then $\calP_\gamma$ defined
  in~\eqref{eq:feasibleset} is non-empty for $\gamma \geq 0$, and
  $f(\gamma)$ defined in~\eqref{eq:minIRUX1} is non-increasing in
  $\gamma$.
\end{proposition}

\begin{proof}
  A non-vacuous feasible set is demonstrated as follows.  Let $(R, U,
  X)$ be a given EPS system. Define a second EPS system $(R', U', X')$
  as follows.  Let $(R', U') = (R, U)$ and $X' = (X, A)$, where $A$ is
  a random variable independent of $(R, U, X)$ such that $H(A) =
  \delta$ for any given $\delta \geq 0$.  In other words, $(R', U',
  X')$ is constructed by adding some spurious randomness into the
  ciphertext of the EPS system $(R, U, X)$.  Setting $\delta = \gamma$
  and supposing that $(R, U, X)$ is a cipher system using a one-time
  pad yields $P_{R'X'|U'} \in \calP_\gamma$.

  By the same trick, we can show that $f(\gamma)$ is non-increasing.
  For any $\gamma > 0$ and $\epsilon > 0$, let $(R, U, X)$ be an EPS
  system such that $P_{RX|U} \in \calP_\gamma$ and
  \begin{equation}
    I(R; UX) < f(\gamma) + \epsilon. \label{eq:minIRUX3}
  \end{equation}
  It is easy to check that $P_{R'X'\given U'} \in \calP_{\gamma+\delta}$ and
  $H(X\given UR) = H(X'\given U'R') - \delta$.  Then
  \begin{align}
    f(\gamma + \delta) &= \inf_{P_{\tR \tX| \tU}\in
      \calP_{\gamma+\delta}} I(\tX; \tU \tR) \\
    &= \inf_{P_{\tR \tX| \tU}\in \calP_{\gamma+\delta}} \left(H(\tX) -
    H(\tX\given \tU \tR )\right) \\
    &= \log |\calU| + \gamma + \delta - \sup_{P_{\tR \tX\given \tU}\in
      \calP_{\gamma+\delta}} H(\tX\given \tU \tR)  \label{eq:minIRUX4}\\
    &\leq \log |U| + \gamma + \delta - H(X'\given U'R')\\
    &= H(X)  - H(X\given UR)\\
    &< f(\gamma) + \epsilon, \label{eq:minIRUX5}
  \end{align}
  where \eqref{eq:minIRUX5} follows from~\eqref{eq:minIRUX3}.  Since
  $\epsilon > 0$ is arbitrary, the second claim of the proposition is
  proved.
\end{proof}

\section{Compression before Encryption \label{se:ComEnc}}
In Section~\ref{se:intro} we discussed the standard approach of
compression-before-encryption (cf. Fig.~\ref{fig:comencry}) suggested by
Shannon. In the following, we will show that this approach is not
necessarily the right way to minimize either $I(R; UX)$ or $H(X)$ in
error free perfect secrecy systems.  For simplicity, all units in this
section are in bits and logarithms are with base $2$.

A central idea in lossless data compression is to encode frequently
occurring symbols (or strings) using shorter codewords. However, this
can cause problems in the context of EPS systems.  For instance,
suppose our cipher consists of a Huffman code followed by a one-time
pad using a key with the same length as the Huffman codeword.  At
first glance, this approach can reduce both the ciphertext size and
the key size to the minimum expected codeword length.  Unfortunately,
this method is not secure because the \emph{length} of the output
discloses some information about the message. Consider an extreme case
that the message is generated according to $P_U(i) = 2^{-i}$ for $1
\leq i < \ell$ and $P_U(\ell) = 2^{-(\ell-1)}$. If a binary Huffman
code is used, the message is uniquely identified by the length when $U
< \ell - 1$.

This problem can be solved by different methods.  One solution has
been discussed in Example~\ref{eg:KeyConsumption}.  In this section,
we only consider the \emph{compress-encrypt-pad scheme} of
Fig.~\ref{fig:PrefixPad}, since this is sufficient to illustrate the
deficiencies of compression before encryption.

In Fig.~\ref{fig:PrefixPad}, a prefix code is used to encode the
message $U$ and a codeword with length $\sigma(U)$ is obtained.  The
codeword is further encrypted by one-time pad using a key with the
same length $\sigma(U)$.  After application of the one-time pad, fair
bits are appended such that the output has a constant length $\gamma$
equal to the longest codeword, $\max_{u \in \calU} \sigma(u)$.  The
receiver decrypts the message by applying the key bit-by-bit to the
ciphertext until a codeword in the prefix code is obtained.

\begin{figure}[htbp]
  \begin{center}
    \includegraphics[scale=0.8]{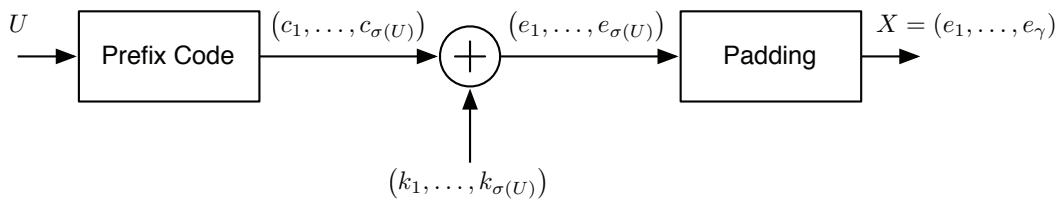}
    \caption{\label{fig:PrefixPad}A compression-encryption-padding scheme.
    }
  \end{center}
\end{figure}

In this scheme, the ciphertext $X$ has a uniform distribution so that
$H(X) = \gamma$.  Since $\gamma$ is the length of the longest codeword
and a prefix code is uniquely decodable, $\gamma \geq \log \ell$,
where $\ell=|\calU|$.  Therefore, $H(X) \geq \log \ell$, in agreement
with Theorem~\ref{th:BoundonXR}.  This scheme requires an initial key
of length $H(R) \geq \log \ell$ bits providing a sufficiently long
secret key in case the longest codeword is the one that happens to be
generated.

Let us now compare the performance of this scheme with the bounds
obtained in Section~\ref{se:sub1MinIRUX}, where the minimal expected
key consumption is assumed.  Suppose the Shannon code~\cite{bk:Cover}
is used in the scheme described in Fig.~\ref{fig:PrefixPad} to
construct an EPS system.  The performance is given in the following
theorem.
\begin{theorem}
  If the Shannon code is used in the compress-encrpyt-pad scheme
  described in Fig.~\ref{fig:PrefixPad} to construct an EPS system,
  then
  \begin{equation}
    H(R) = H(X) = \left\lceil \log \frac{1}{\pi} \right\rceil, \label{eq:ecl-2}
  \end{equation}
  which exceeds the lower bound in~\eqref{eq:minIXR0HXHR2} by no
  more than $1$ bit. Furthermore, the expected key consumption exceeds
  the lower bound~\eqref{eq:LBonEKC} by no more than 1 bit,
  \begin{equation*}
    I(R; UX) \leq H(U) +1.
  \end{equation*}
\end{theorem}
\begin{proof}
  Recall that $\sigma(u)$ is the length of the codeword assigned to $U
  = u$.  Then the longest codeword has length equal to
  \begin{align}
    \left\lceil \log \frac{1}{\pi} \right\rceil,
  \end{align}
  where $\pi = \min_{u\in\calU } P_U(u)$.  Recall in
  Fig.~\ref{fig:PrefixPad} that fair bits are appended to each
  codeword to construct a constant length ciphertext $X$.  Therefore,
  $H(R) = H(X) = \left\lceil \log \frac{1}{\pi} \right\rceil$ which is
  within one bit of the lower bound in \eqref{eq:minIXR0HXHR2}.
  Furthermore, the expected key consumption
  \begin{align}
    I(R; UX)
    &= H(R) - H(R\given UX)  \label{eq:ecl-1}\\
    &= \left\lceil \log \frac{1}{\pi} \right\rceil - \left(\left\lceil \log \frac{1}{\pi} \right\rceil - \sum_{u \in \calU} P_U(u) \sigma(u) \right)  \label{eq:ecl0}\\
    &= \sum_{u \in \calU} P_U(u) \sigma(u)   \label{eq:ecl}\\
    &\leq H(U) +1, \label{eq:ecl2}
  \end{align}
  where~\eqref{eq:ecl0} follows from the fact that $H(R\given UX)$ is
  equal to the number of appended fair bits, and~\eqref{eq:ecl2}
  follows from \cite[(5.29)--(5.32)]{bk:Cover}.  Therefore, $I(R; UX)$
  is also within a bit of the lower bound in \eqref{eq:LBonEKC}.
\end{proof}
Therefore, we conclude that if the Shannon code is used for
compression in Fig.~\ref{fig:PrefixPad}, then the performance is close
to the optimal code in the minimal key consumption regime when both
$H(U) \gg 1$ and $ \log \frac{1}{\pi} \gg 1$.

Now, we compare the performance obtained when the Huffman code is used
in place of the Shannon code.  In this case, the expected key
consumption $I(R; UX)$ can again be analyzed similar
to~\eqref{eq:ecl-1} -- \eqref{eq:ecl2}.  Since the expected codeword
length in~\eqref{eq:ecl} is shorter for the Huffman code, a smaller
$I(R; UX)$ can be obtained.  However, the longest codeword in the
Huffman code can be longer than the longest codeword in the Shannon
code.  As a consequence, larger $H(X)$ and $H(R)$ are required for
certain $P_U$.  This can be seen in the example in
Table~\ref{table:1}.  In the worst case, the longest codeword in the
Huffman code can be as much as $44\%$ longer than the longest codeword
in the Shannon code~\cite{abu2000maximal}.  Furthermore, the partition
code $\calC(\Psi)$ in Table~\ref{table:1} outperforms the compression
before encryption schemes based on either the Huffman code or the
Shannon code because $\calC(\Psi)$ is optimal according to
Theorem~\ref{th:optimalPart}.  On the other hand, the Shannon code
uses unnecessarily long codewords for certain source distributions,
e.g., $P_U = (0.9, 0.1)$.  As a consequence, larger $H(X)$ is needed
as shown in Table~\ref{table:2}.  However, the minimal $I(R; UX)$ or
the minimal $H(X)$ can be obtained using different partition codes.
We conclude that compression before encryption is a suboptimal
strategy to minimize key consumption or the number of channel uses in
EPS systems.

\begin{table}[htbp] \center
\caption{Comparing different schemes with $\Phi = (1, 1, 1, 3, 4, 7, 11)$ and $P_U(i) = {\Phi(i)}/{28}$ for $1 \leq i \leq 7$}
\label{table:1}
{
\begin{tabular}{|c|c|c|c|}
\hline
 & Huffman & Shannon & Partition $\calC(\Phi)$\\
\hline
\hline
$I(R; UX)$ & $2.357$ & $2.679$ & $2.291 = H(U)$
\\
\hline
$H(X)$  & $6$ & $5$ & $5$\\
\hline
\end{tabular}
}
\end{table}
\begin{table}[htbp] \center
\caption{Comparing different schemes with $\Phi = (9, 1)$, $\Phi' = (1, 1)$ and $P_U = (0.9, 0.1)$}
\label{table:2}
{
\begin{tabular}{|c|c|c|c|c|}
\hline
 & Huffman & Shannon & Partition $\calC(\Phi)$ & Partition $\calC(\Phi')$\\
\hline
\hline
$I(R; UX)$ & $1$ & $1.3$ & $0.469 = H(U)$ & $1$
\\
\hline
$H(X)$  & $1$ & $4$ & $4$ & $1$\\
\hline
\end{tabular}
}
\end{table}

Suppose now that the source distribution is $d$-adic and the smallest
probability mass in $P_U$ is equal to $d^{-{k}}$ for certain integers
$d$ and ${k}$.  What were binary digits in the scheme described above in
Fig.~\ref{fig:PrefixPad} now become $d$-ary symbols.  It can be
verified that the longest codeword has length equal to ${k}$.
Therefore, both $d$-ary Shannon codes and $d$-ary Huffman codes can
achieve the minimal $H(X)$ and $H(R)$ in~\eqref{eq:minIXR0HXHR2}.
Furthermore, the expected codeword length is equal to $H(U)$.
By~\eqref{eq:ecl}, $I(R; UX)$ is equal to the expected codeword
length, which is equal to $H(U)$. Therefore, the minimal $I(R; UX)$ is
achieved.  However, a prefix code cannot achieve the expected codeword
length $H(U)$ when $P_U$ is not
$d$-adic~\cite[Theorem~4.6]{bk:Raymond}.  Again consider the example
in Table~\ref{table:2} where $P_U = (0.9, 0.1)$.  Only partition code
but neither the Shannon nor the Huffman code can be used to achieve
$I(R; UX) = H(U)$.  Indeed, the $d$-adic distribution is just a
special case of the condition used in Theorem~\ref{th:optimalPart}.
Therefore, the partition code can achieve the minimal $I(R; UX)$ for a
wider range of $P_U$.

\section{Proof of Theorem~\ref{th:supportXR}}
\label{se:supportXR}
Suppose there exists $u \in \calU$ such that $P_{U}(u)$ is
irrational. Define a new random variable $U^{*}$ such that
\begin{equation*}
  U^{*} =
  \begin{cases}
    0 & \text{ if } U = u \\
    1 & \text{ otherwise.}
  \end{cases}
\end{equation*}
Then $P_{U^*}(0)$ and $P_{U^*}(1)$ are irrational.  As $U^{*}$ is a
function of $U$, by~\eqref{eq:IUX0} -- \eqref{eq:IUR0}
and~\eqref{eq:IXR0},
\begin{align}
  I(U^{*};R) = I(U^{*};X) = I(X;R) = H(U^{*}\given XR) = 0. \label{eq:4.3}
\end{align}
Therefore, it suffices to consider binary $U$.

Let $\calX$ and $\calR$ be the respective supports of $X$ and $R$.
Suppose to the contrary first that $|\calX| $ and $|\calR|$ are both
finite.  We can assume without loss of generality that
\begin{align}
  \calX & = \{1,\ldots, n\}\\
  \calR & = \{1,\ldots, m\}.
\end{align}
Let
\begin{align}
  x_{i} &=  P_{X}(i), \quad  i=1,\ldots, n\\
  r_{j} &= P_{R}(j), \quad j=1,\ldots, m,
\end{align}
and let $\bf x$ be the $n$-row vector with entries $x_{i}$. Similarly,
define the column vector $\bf r$.

As $X$ and $R$ are independent and $H(U\given XR)=0$, there exists a
function $g$ such that $U= g(X,R)$. Hence, from $X$ and $R$ we induce
a $n \times m$ \emph{decoding matrix} $G$ with entries
\begin{equation*}
  G_{i,j}=f(i,j),\quad i = 1,\ldots, n,\, j=1,\ldots, m.
\end{equation*}
Then
\begin{align}
   \sum_{j=1}^{m}  G_{i,j} r_{j} &= P_{U}(1),  \quad
  i=1,\ldots, n \label{eq:G1} \\
  \sum_{i=1}^{n}x_{i} &= \sum_{j=1}^{m}r_{j} = 1 \label{eq:G2}\\
  x_{i} &\geq 0, r_j \geq 0,\quad i=1,\dots,n,\,j=1,\dots,m \label{eq:G3}\\
  \sum_{i=1}^{m} x_{i} G_{i,j} &= P_{U}(1), \quad j=1,\dots,m\label{eq:G4}
\end{align}
Here,~\eqref{eq:G1} is due to the fact that $I(U;X)=0$,~\eqref{eq:G2}
and~\eqref{eq:G3} are required since $P_{X}$ and $P_{R}$ are
probability distributions, and~\eqref{eq:G4} follows from $I(U;R) =
0$.

In fact, for any $\bf x$, $\bf r$ and binary matrix $G$ satisfying the
above four conditions, one can construct random variables $\{U,R,X\}$
such that
\begin{align}
  I(U ;R) = I(U ;X) = I(X;R) = H(U \given XR)  = 0
\end{align}
where $U=f(X,R)$ and the probability distributions of $X$ and $R$ are
specified by the vectors $\bf x$ and $\bf r$ respectively.

In the following, we will prove that if the rows of $G$ are not
independent, then we can construct another random variable $X^{*}$
with support $\calX^*$,  $|{\calX^*}| < |\calX|$
such that
\begin{align}
  I(U ;R) = I(U ;X^{*}) = I(X^{*};R) = H(U\given X^{*}R) = 0.
\end{align}

\def\A{{\calA}} \def\B{{\calB}}

To prove this claim, suppose that there exists disjoint subsets $\A$
and $\B$ of $\{1,\ldots, n\}$ and positive numbers $\alpha_{i}, i\in\A
\cup\B $ such that
\begin{align}\label{eq:88}
\sum_{i\in\A} \alpha_{i}G_{i} = \sum_{k\in\B} \alpha_{k}G_{k}.
\end{align}
where $G_{i}$ is row $i$ of $G$.  Then we will claim that
\begin{equation*}
  \sum_{i\in\A} \alpha_{i} = \sum_{k\in\B}\alpha_{k}.
\end{equation*}
Multiplying both sides of~\eqref{eq:88} by ${\bf r}$,
\begin{align}
  \sum_{i\in\A} \alpha_{i}G_{i} {\bf r}  & = \sum_{k\in\B} \alpha_{k}G_{k}{\bf r} \\
  \sum_{i\in\A} \alpha_{i}P_{U}(1)  & = \sum_{k\in\B} \alpha_{k}P_{U}(1)\\
  \sum_{i\in\A} \alpha_{i}   & = \sum_{k\in\B} \alpha_{k}.
\end{align}
Let $\epsilon \triangleq \min_{i \in \A \cup \B}
{x_{i}}/{\alpha_{i}}$.  Assume without loss of generality that $n\in
\A$ and that $\epsilon = {x_{n}}/{\alpha_{n}}$.
Define
\begin{equation*}
  x_{i}^{*} =
  \begin{cases}
    x_{i} - \epsilon \alpha_{i}, & i\in\A \\
    x_{i} + \epsilon \alpha_{i}, & i\in\B \\
    x_{i},& \text{ otherwise. }
  \end{cases}
\end{equation*}
Note that $x^{*}_{n}=0$. Suppose that the probability distribution of
$X$ is changed such that $P_{X}(i) = x^{*}_{i}$.  Then it can be
checked easily that $U,X,R$ still satisfy~\eqref{eq:IUX0} --
\eqref{eq:IUR0} and~\eqref{eq:IXR0}.  Furthermore, the size of the
support of $P_{X}$ is $|\calX| \le n-1 $.

Repeating this procedure, we can prove that for any random variable
$U$, if there exists auxiliary random variables $X,R$
satisfying~\eqref{eq:IUX0} -- \eqref{eq:IUR0} and~\eqref{eq:IXR0},
then there exists  auxiliary random variables
$X^{*},R^{*}$ such that~\eqref{eq:4.3} is satisfied and the rows
and columns of the decoding matrix induced by $X^{*}$ and $R^{*}$ are
all linearly independent.  Hence, the decoding matrix $G$ induced by
$X^{*}$ and $R^{*}$ must be square (and thus $m=n$).  Consequently,
\begin{align}
  \sum_{i=1}^{n} x_{i} G_{i,j} & = P_{U}(1) ,\quad j=1,\ldots, n .
\end{align}
There exists a unique solution $(z_{1},\ldots, z_{n})$
such that
\begin{align}
  \sum_{i=1}^{n} z_{i} G_{i,j} & = 1 ,\quad j=1,\ldots, n.
\end{align}
Clearly, $z_{i}= x_{i} / P_{U}(1)$.  As all the entries in $G$ are
either $0$ or $1$, all the $z_{i}$ are rational numbers.  Therefore,
\begin{align}
1 & = \sum_{i=1}^{n}  x_{i}   = P_{U}(1) \sum_{i=1}^{n}  z_{i}.
\end{align}
Hence, $P_{U}(1)$ must be rational and a contradiction occurs.  We
have proved that $\calX$ and $\calR$ cannot be both finite.  The case
when only $\calX$ or $\calR$ is finite can be similarly proved.

\section{Conclusion}
This paper studied perfect secrecy systems with zero decoding error at
the receiver, with the additional assumption that the message $U$ and the
secret key $R$ are independent, $I(U; R) = 0$.  Under this setup, we found
a new bound $\log |\calU| \leq H(R)$ on the key requirement, improving
on Shannon's fundamental bound $H(U) \leq H(R)$ for perfect secrecy.

To transmit the ciphertext $X$, the lower bound on the minimum number
of channel uses has been shown to be $\log |\calU| \leq H(X)$.  If the
source distribution is defined on a countably infinite support or a
support with unbounded size, no security system can simultaneously
achieve perfect secrecy and zero decoding error.

We also defined and justified three new concepts: \emph{residual
  secret randomness}, \emph{expected key consumption}, and
\emph{excess key consumption}.  We have demonstrated the feasibility
of extracting residual secret randomness in multi-round secure
communications which use a sequence of error free perfect secrecy
systems.  We quantified the residual secret randomness as $H(R|UX)$.  We further distinguished between the size $H(R)$ of the secret
key required prior to the commencement of transmission, and the
expected key consumption $I(R; UX)$ in a multi-round setting.  In
contrast to $H(R)\geq \log|\calU|$, we showed that $I(R; UX)$ is lower
bounded by $H(U)$, giving a more precise understanding about the role
of source entropy in error free perfect secrecy systems.  The excess
key consumption is quantified as $I(R; X)$, and is equal to $0$ if
and only if the minimal expected key consumption is achieved.

One of the main objectives of this paper was to reveal the fundamental
tradeoff between expected key consumption and the number of channel
uses.  For the regime where the minimal $I(R; UX)$ is assumed, $H(X)$
and $H(R)$ are inevitably larger and corresponding lower bounds for
$H(X)$ and $H(R)$ have been obtained.  If the source distribution
$P_U$ has irrational numbers, the additional requirements on the
alphabet sizes of $X$ and $R$ to achieve minimal $I(R; UX)$ have been
shown.  We have proposed a new code, the \emph{partition code}, which
generalizes the one-time pad, and can achieve minimal $I(R; UX)$ when
all the probability masses in $P_U$ are rational.  In some cases, the
partition code can simultaneously attain the minimal $H(X)$ and $H(R)$
in this regime.

At the other extreme, the regime where the minimal number of channel
uses is assumed, the one-time pad has been shown to be optimal.  For
the intermediate regime, we have formulated an optimization problem
for the fundamental tradeoff between $I(R; UX)$ and $H(X)$.  We also
demonstrated that compression before encryption cannot minimize either
$H(R)$, $H(X)$ or $I(R; UX)$.

This paper has highlighted a few open problems.  First, the complete
characterization of the tradeoff between $I(R; UX)$ and $H(X)$ remains
open.  Second, the partition code is only one class of codes designed
to minimize expected key consumption. Codes achieving other points on
the tradeoff curve are yet to be discovered.  In particular, a code
achieving minimal $H(X)$ and $H(R)$ in the regime of minimal expected
key consumption is important for the design of efficient and secure
systems.

%

\end{document}